\documentclass[pdflatex,sn-mathphys-num]{sn-jnl}
\usepackage{graphicx}
\usepackage{multirow}
\usepackage{amsmath,amssymb,amsfonts}
\usepackage{amsthm}
\usepackage{enumerate}
\usepackage[title]{appendix}
\usepackage{xcolor}
\usepackage{textcomp}
\usepackage{manyfoot}
\usepackage{booktabs}
\usepackage{algorithm}
\usepackage{algorithmicx}
\usepackage{algpseudocode}
\usepackage{listings}
\usepackage{bussproofs}
\usepackage{framed}
\usepackage{array}
\usepackage[normalem]{ulem}
\usepackage{xspace}

\theoremstyle{thmstyleone}
\newtheorem{theorem}{Theorem}
\newtheorem{proposition}[theorem]{Proposition}
\newtheorem{lemma}[theorem]{Lemma}

\theoremstyle{thmstyletwo}
\newtheorem{example}{Example}

\theoremstyle{thmstylethree}

\renewcommand{\iff}{\ \text{iff} \ }
\newcommand{\Lconst}[2]{%
\ifthenelse{\equal{#1}{*}}{%
\settowidth{\tlen}{\ensuremath{\text{start}}}\hbox
  to \tlen{\hfil\ensuremath{\text{#2}}}}{%
\text{#2}}}

\newcommand{\bi}{\begin{itemlist}}
\newcommand{\ei}{\end{itemlist}}
\newcommand{\bdl}{\begin{descriptionlist}}
\newcommand{\edl}{\end{descriptionlist}}
\newcommand{\bthm}{\begin{theorem}}
\newcommand{\ethm}{\end{theorem}}
\newcommand{\bcor}{\begin{corollary}}
\newcommand{\ecor}{\end{corollary}}

\newcommand{\mlimpl}{\mathrel{\rightarrow}}
\newcommand{\psubset}{\ensuremath{\mathrel{\subseteq}}}

\newcommand{\mlequiv}{\mathrel{\equiv}}

\newcommand{\mland}{\mathrel{\wedge}}
\newcommand{\mlor}{\mathrel{\vee}}

\renewcommand{\iff}{\ \text{iff} \ }

\DeclareMathOperator{\Cls}{Cls}

\newcommand{\ml}[1]{\ensuremath{\mathit{#1}}}

\newenvironment{descriptionlist}[1]%
        {\begin{list}{}{
                \setlength{\topsep}{3pt}
                \setlength{\itemsep}{3pt}
                \setlength{\parsep}{0pt}
                \settowidth{\labelwidth}{#1}
                \settowidth{\leftmargin}{#1\hspace{\labelsep}}}
        }{
        \end{list} }
\makeatletter
\newenvironment{itemlist}{%
  \ifnum \@itemdepth >\thr@@\@toodeep\else
    \advance\@itemdepth\@ne
    \edef\@itemitem{labelitem\romannumeral\the\@itemdepth}%
    \expandafter
    \begin{list}{\csname\@itemitem\endcsname}{%
                \setlength{\topsep}{3pt}
                \setlength{\itemsep}{3pt}
                \setlength{\parsep}{0pt}
                \settowidth{\labelwidth}{\csname\@itemitem\endcsname}
                \settowidth{\leftmargin}{\csname\@itemitem\endcsname\hspace{\labelsep}}%
                }\fi
        }{%
        \end{list} }
\makeatother

\newcommand{\BPL}{\ensuremath{BPL}\xspace}

\newcommand{\BPLplus}{\ensuremath{\BPL_{fcp}^{\not\approx}}\xspace}

\newcommand{\cresol}{\ensuremath{\mathit{\mathsf{C}}}\xspace}
\newcommand{\CSk}[1]{\mathop{\mathrm{CSk}(#1)}}
\DeclareMathOperator{\sig}{sig}

\newcommand{\foeq}{\mathrel{\approx}}
\newcommand{\foneq}{\mathrel{\not\approx}}

\newcommand{\ba}{\begin{array}}
\newcommand{\ea}{\end{array}}

\usepackage{tikz}

\begin{document}

\title{Second-Order Quantifier Elimination and Uniform Interpolation
for Bath Path Logic and the Ordered Fragment}

\author*[1]{\fnm{Renate A.} \sur{Schmidt}}

\author[2]{\fnm{Hongkai} \sur{Yin}}

\affil*[1]{\orgdiv{Department of Computer Science}, \orgname{The University of Manchester}, \orgaddress{\street{Oxford Road}, \city{Manchester}, \postcode{M13 9PL}, \country{UK}}}

\affil[2]{\orgdiv{Department of Philosophy}, \orgname{Central European University}, \orgaddress{\street{Quellenstraße 51}, \city{Wien}, \postcode{1100}, \country{Austria}}}

\abstract{%
We consider and extend results on basic path logic and
the ordered fragment of first-order logic, both of which originate from
the functional translation of modal logic.
Basic path logic is a subclass of the clausal class for the $\exists^*\forall^*$-fragment 
and has the remarkable property that binary resolution decides it.
This decidability result and the consequence finding completeness
of binary resolution allows us to observe that binary resolution
also decides uniform interpolation and computes uniform interpolants
for basic path logic.
By introducing constant Skolemisation, we show that sentences of the ordered fragment can be transformed into basic path logic, and this transformation preserves logical consequences in the ordered fragment. 
We characterise the search space of the SCAN algorithm on the clausal form of the ordered fragment by a variation of basic path logic and prove that SCAN terminates on this class,
and therefore it decides second-order quantifier elimination
for this class.
It remains unclear whether uniform interpolants in the ordered fragment can be extracted from the output of SCAN.
}


\maketitle

\section{Introduction}

In this paper we investigate the application of the SCAN
algorithm~\cite{GabbayOhlbach92a} and resolution in solving
second-order quantifier elimination and computing uniform
interpolants for basic path logic~\cite{OhlbachSchmidt97,Schmidt99a}.

We say that a clause is \emph{prefix stable for variables} if for any \emph{variable} in the clause, its prefix in any argument sequence (in which it occurs) is the same~\cite{Ohlbach88b}. For example, the first clause below is prefix stable for variables and the second is not (where $y$ has two different prefixes):
\begin{align*}
    &P(x,y)\lor Q(x,y,a)\\
    &P(x,y)\lor Q(a,y,z).
\end{align*}
By \emph{basic path logic} (\BPL) we understand the clausal class in which:
(i) all argument sequences contain only variables and constants, and 
(ii) all clauses have prefix stability for variables. 
This implies \BPL is a subclass of the clausal class for the $\exists^*\forall^*$-fragment of first-order logic.

Basic path logic has been introduced as the clausal class
associated with modal logic $\ml{KD}$ based on the world path
semantics and the functional translation to first-order logic~\cite{OhlbachSchmidt97,Schmidt99a}.
Basic path logic is computationally well-behaved:
it is decidable by resolution without special refinements 
(no ordering- or selection-based restriction of inferences
is needed)~\cite{Schmidt99a}, 
which
implies that any sound and refutationally complete
ordering- and/or selection-based refinement provides a decision
procedure for this class.\footnote{%
On the side we remark, these observations do not extend to the
$\exists^*\forall^*$ clausal class which is NEXPTIME-complete, but
with a range-restriction transformation hyper-resolution
decides the $\exists^*\forall^*$ clausal class~\cite{BaumgartnerSchmidt20}.
}
Further, the performance of resolution theorem provers is superior
for modal
logic problems when they are mapped to basic path logic than when
they are mapped to the guarded fragment using the standard relational
translation~\cite{HustadtSchmidt97b,HustadtSchmidt99e}.

The \emph{ordered fragment} of first-order logic consists of those formulas in which, roughly speaking: (i) the argument sequence of each predicate is of the form $(x_1,x_2,\dots,x_n)$; and (ii)~if $i<j$, no quantifier binding $x_j$ out-scopes any quantifier binding $x_i$. For example, the following sentence belongs to the fragment:
\[\forall x_1\exists x_2(R(x_1,x_2)\land \forall x_3\neg S(x_1,x_2,x_3)).\]
The ordered fragment first appeared in the ``fluted schema'' introduced by W. V. Quine \cite{Quine1969,Quine1975}, but the formal definition we follow was given, independently of Quine's work, by A.~Herzig~\cite{Herzig1990}, where the term ``ordered formula'' was coined. 
Herzig proposed a linear-time reduction from satisfiability of ordered formulas to satisfiability in $\ml{KD}$, which implies that the former is decidable in PSPACE.
Also, it is proved in \cite{OhlbachSchmidt97} that the quantifier prefix in the prenex normal form of any ordered formula can be freely permuted without affecting satisfiability; thus, by this permutability, every ordered formula can be transformed into \BPL \cite{Schmidt97d,Schmidt99a}. Moreover, the ordered fragment is shown in~\cite{Bednarczyk2022} to have the Craig Interpolation Property.

\emph{Second-order quantifier elimination
(SOQE)} is the problem of reducing a
formula $\exists \overline X \varphi$, where~$\overline X$ are predicate
variables and $\varphi$ is a first-order formula, to an equivalent first-order
formula $\psi$ (possibly with equality) that does not contain any of the predicate
symbols~$\overline X$~\cite{GabbayOhlbach92a,GabbaySchmidtSzalas08}. For example,
\[
\exists X (Xa\land\forall x(Xx\to Bx)) \equiv Ba \quad \text{and} \quad \exists X(Xa\land\neg Xb) \equiv a\not\approx b.
\]

SOQE is an important topic of artificial intelligence
that has real-world applications in diverse areas, from knowledge
representation, knowledge sharing and communication of agents, to
answer set programming and logic, including description logics and
modal correspondence theory.
SOQE has been studied in the guise of forgetting~\cite{LinReiter94}, 
strongest necessary and weakest sufficient conditions in knowledge representation~\cite{KonevWaltherWolter09,KoopmannSchmidt13a,KoopmannSchmidt13c,Delgrande17,FereeVanDerGiessenEtAl24,DohertyLukaszewiczSzalas01},
symbol elimination
and projection~\cite{PeuterSofronieStokkermans21}.
SOQE has applications in knowledge representation for ontology
extraction, ontology comparison,
and abductive reasoning~\cite{KoopmannSchmidt13a,KoopmannSchmidt13c,ChenAlghamdiEtAl19a,LudwigKonev14,LiuLuEtAl21,DelPintoSchmidt19a,KoopmannDelPintoEtAl20},
and in artificial intelligence for
common-sense reasoning and circumscription~\cite{GabbayOhlbach92a},
databases~\cite{GabbaySchmidtSzalas08} and answer set programming
~\cite{EiterKernIsberner19,GoncalvesKnorrLeite23}. 

SOQE is tightly connected to the problem of computing \emph{uniform
interpolants}.
Given a first-order formula $\varphi$ and a subset $\Sigma$ of
$\sig(\varphi)$, the signature of~$\varphi$, a uniform interpolant
of $\varphi$ with regard to $\Sigma$ is a first-order formula $\psi$
such that $\sig(\psi)=\Sigma$ and
\[
\models\varphi\to\chi \quad \iff \quad \models\psi\to\chi \ \text{for every $\chi$ with $\sig(\varphi)\cap \sig(\chi)\subseteq \Sigma$}.
\]
Note that if $\overline{X}=\sig(\varphi)\setminus\Sigma$
and $\psi$ is the result of SOQE for $\varphi$, i.e.,
$\psi\equiv\exists\overline{X}\varphi$, then $\psi$ is a uniform
interpolant of $\varphi$ with regard to $\Sigma$.
For example, $Ba$ is a uniform interpolant of 
$Xa \wedge \forall x (Xx \rightarrow Bx)$ with regard to the signature containing only $B$ and $a$,
and $a\not\approx b$ is a uniform interpolant of $Xa \wedge \neg Xb$ with regard to the signature containing only $a$ and $b$.
Also, when we are interested in uniform interpolation in a fragment of first-order logic, the aforementioned~$\varphi$,~$\psi$ and $\chi$ are assumed to be formulas of the fragment.
In other words, the construction of uniform interpolants is dependent on the target language.
For example, a uniform interpolant of $Xa \wedge \neg Xb$
is $\top$ in first-order logic \emph{without} equality (instead of $a\not\approx b$).
However, the target language of SOQE is always assumed to be first-order logic \emph{with} equality.

There are few settings in which SOQE and uniform interpolation are known to be decidable.
For propositional logic SOQE and uniform interpolation are decidable,
but in more general settings, e.g., for many modal logics, description
logics and first-order logic, they are not even semi-decidable.

Efforts to automate modal correspondence theory led to the development
of the SOQE algorithm SCAN for first-order clause
logic~\cite{GabbayOhlbach92a,Engel96}.
Given a clause set $N$ and predicate variables $\overline X$, the SCAN algorithm attempts to eliminate the symbols
$\overline X$ in an iterative process performing all necessary
inferences upon $X$-literals and deleting clauses once they are no
longer needed.
On successful termination, the returned clause set $M$ is equivalent to
$\exists \overline X N$.
The soundness of SCAN for the SOQE problem has been proved in \cite{GabbayOhlbach92a}, 
but because of the inherent undecidability of SOQE it cannot generally be complete.
It was however shown that SCAN is complete for solving the correspondence
problem of the class of Sahlqvist axioms in modal
logic~\cite{GorankoHustadtSchmidtVakarelov04a}.
It is also known to be sound and complete for propositional problems.
Moreover, if SCAN does not terminate then the conjunction of the infinite
set of persisting clauses (closed under first-order universal
quantification and) free of predicate variables from $\overline X$ is
a solution to the problem~\cite{GabbayOhlbach92a}.

Because of the mentioned nice properties of basic path logic in resolution, the question arises
whether the SCAN algorithm can be applied to solve the
SOQE problem and compute uniform interpolants in this case.
The main research questions and contributions of this paper are the following:

\paragraph{1. Characterisation of the search space of SCAN on basic path logic with fixed constant positions.}
The SCAN algorithm based on constraint resolution replaces unification in the resolution and factoring rules by adding inequalities
(constraints) to conclusions.
These constraints fall outside the scope of the language of basic path
logic, which means the latter is insufficient to characterize the search
space of SCAN. 
For that reason,
and because \BPL is slightly more general than it needs to be for
deciding the ordered fragment and modal logic \ml{KD},
we introduce the class \BPLplus, 
in which inequalities are allowed and constants have fixed argument positions, and
show that \BPLplus is closed under the rules of SCAN. 

\paragraph{2. Termination and completeness of SCAN on \BPLplus.}
We further show that 
the set of \BPLplus-clauses over a finite relational signature is finitely
bounded modulo condensing (a form of subsumption deletion).
This implies that SCAN is a decision procedure for SOQE
and is SOQE-complete for \BPLplus.
The result can be seen to follow from the characterisation and
properties of \BPL in~\cite{Schmidt97d,Schmidt99a}, noting that
normalised conclusions of SCAN inferences are either standard
resolvents or conditionals enumerating the coincidence of pairs of
constants with the same indices.

\paragraph{3.
Consequence-preserving clausification for the ordered fragment.}
The previous transformation of the ordered fragment to clauses in \BPL or \BPLplus
involves an unusual quantifier exchange operation \cite{OhlbachSchmidt97}.
Although it is sufficient for solving the satisfiability problem,
for uniform interpolation we need a clausification that preserves
logical consequences in the ordered fragment.
To address this problem, we introduce constant Skolemisation {(CSk)}
and show that: for any sentences $\phi$ and $\psi$ of the ordered
fragment,
$\models \phi \mlimpl \psi$ iff  $\models \CSk \phi \mlimpl \psi$. Moreover, the clausification of ordered sentences using CSk also results in \BPL and \BPLplus clauses.

\paragraph{4.\ Uniform interpolants (for the ordered fragment) in basic path logic.}
Unskolemisation of a clause set returned by SCAN is always possible since the clauses in question contain only variables and constants.
Specifically, every such clause set can be expressed as a first-order formula in prenex normal form with a $\exists^*\forall^*$
quantifier prefix. However, such a formula is in general not expressible in the ordered fragment. As a partial solution to the problem, we show that clauses with inequalities between constants can be deleted without any cost on logical consequences in equality-free first-order logic, and therefore we obtain a ``uniform interpolant'' \emph{expressed in basic path logic}. It remains open whether an additional procedure can be devised to construct a uniform interpolant \emph{in} the ordered fragment.

\medskip
The rest of the paper is arranged as follows. In Section~\ref{section_BPL_resolution} we recap the definition of basic path logic and its nice properties when using resolution. These results allow us to conclude uniform interpolation completeness for \BPL. In Section~\ref{section_OF_CSk}, we recap the definition of the ordered fragment, introduce constant Skolemisation, and prove that it preserves logical consequences in the ordered fragment. We present the SCAN algorithm and some basic results in Section~\ref{section_SCAN_algorithm}, and then formulate the clausal class \BPLplus in Section~\ref{section_BPLplus}, where we prove its closure under constraint resolution and the termination and SOQE-completeness of SCAN. Section~\ref{section_UI} discusses the elimination of inequalities and the remaining problem in computing uniform interpolants for the ordered fragment.

\section{Basic Path Logic, Resolution and Uniform Interpolation}
\label{section_BPL_resolution}

Let $u_i$ denote the $i$th element in the sequence $\overline u$ of
terms.
The preﬁx of $u_i$ in $\overline u$ is $\varepsilon$ (the empty sequence) if $i = 1$,
and $(u_1, \ldots, u_{i-1})$, otherwise.

Let $T$ denote a set of sequences of terms.
We say $T$ is \emph{preﬁx stable}, if for any term~$t$, all its
occurrences in $T$ have the same prefix.
We say a clause is \emph{prefix stable} if the set of argument
sequences of the literals is prefix stable.

In \emph{basic path logic (\BPL)} only variables are required to be prefix stable. A set of sequences of terms is \emph{preﬁx stable for variables} if every occurrence of a \emph{variable} $u$ in the set has the same prefix~\cite{Schmidt97d,Schmidt99a}. Recall the example we saw above:
\begin{align*}
    &P(x,y)\lor Q(x,y,a)\\
    &P(x,y)\lor Q(a,y,z).
\end{align*}
The first clause is prefix stable for variables because every occurrence of $x$ has the prefix $\varepsilon$, and every occurrence of $y$ has the prefix $(x)$; by contrast, the second is not prefix stable for variables because the first occurrence of $y$ has the prefix $(x)$ while the second has the prefix $(a)$.

Let $\Sigma$ be a signature of finitely many predicate symbols $P,
Q, \ldots$, constants $a,b, \ldots$ and infinitely many first-order
variables $x, y, \ldots$
By definition, \BPL is the class of $\Sigma$-clauses which are prefix stable for variables.  
Specifically, for any clause $C$ in~\BPL: 
\begin{enumerate}[(i)]
\item
$C$ is a disjunction of literals of the form $(\neg)P(\overline u)$,
where each term $u_i$ in $\overline u$ is either a variable or a
constant, and
\item
for any two literals $(\neg)P(u_1, \ldots, u_m)$ and $(\neg)Q(v_1,
\ldots, v_n)$ in $C$ both the following conditions hold for variables:
\begin{description}
\item[\textbf{T1}]
If $u_i$ and $v_j$ are the same variable then $i=j$.
\item[\textbf{T2}] 
Each pair $u_k$ and $v_k$ preceding $u_i$ and $v_i$
are identical.
\end{description}
\end{enumerate}

\begin{table}[tb]
\caption{The standard inference rules of first-order resolution systems.}
\label{table_resol_inference_rules}
\begin{framed}
\begin{tabular}{@{}p{1\linewidth}@{}}
\textbf{Binary resolution}\qquad
\AxiomC{$C\lor A$}
\AxiomC{$D\lor \neg B$}
\BinaryInfC{$(C\lor D) \sigma$}
\DisplayProof
\end{tabular}
provided $\sigma$ is a most general unifier of atoms $A$ and $B$,
the two premises have no
individual variables in common, and the premises are distinct (to prevent
self-resolution).
\bigskip
\par
\begin{tabular}{@{}p{1\linewidth}@{}}
\textbf{Factoring}\qquad
\AxiomC{$C\lor (\neg)A\lor(\neg)B$}
\UnaryInfC{$(C\lor (\neg)A)\sigma$}
\DisplayProof
\end{tabular}
provided $\sigma$ is a most general unifier of atoms $A$ and $B$.
\end{framed}
\end{table}

\begin{table}[tb]
\caption{Standard redundancy elimination rules.}
\label{table_simplification_rules}
\begin{framed}
\begin{tabular}{@{}p{1\linewidth}@{}}
\textbf{Tautology deletion}\qquad
\AxiomC{$N\cup \{C\}$}
\UnaryInfC{$N$}
\DisplayProof
\end{tabular}
if $C$ is a~tautology.
\bigskip
\par
\begin{tabular}{@{}p{1\linewidth}@{}}
\textbf{Subsumption deletion}\qquad
\AxiomC{$N\cup \{C,D\}$}
\UnaryInfC{$N\cup\{C\}$}
\DisplayProof
\end{tabular}
if $C$ subsumes $D$, i.e.,\ there is a~substitution $\sigma$ such that
$C\sigma \subseteq D$.
\bigskip
\par
\begin{tabular}{@{}p{1\linewidth}@{}}
\textbf{Condensation}\qquad
\AxiomC{$N\cup \{C\}$}
\UnaryInfC{$N\cup\{\mathrm{cond}(C)\}$}
\DisplayProof
\end{tabular}
where $\mathrm{cond}(C)$ is a~minimal
(with respect to the number of literals) subclause $D$ of $C$ such
that there exists a~substitution $\sigma$
with $L\sigma\in D$ for every $L\in C$.
\end{framed}
\end{table}

\BPL has the property that prefix
stability for variables is preserved by binary resolution and factoring
defined in Table~\ref{table_resol_inference_rules} and also 
tautology deletion, subsumption deletion and
condensation~\cite{Schmidt97d,Schmidt99a}.
This implies

\begin{theorem}
\label{theorem_BPL_closed_under_resol}
\BPL is closed under resolution systems based on the rules in
Tables~\ref{table_resol_inference_rules}
and~\ref{table_simplification_rules}.
\end{theorem}

\begin{theorem}[\cite{Schmidt97d,Schmidt99a}]
\label{theorem_resol_decides_BPL}
Let $N$ be any finite set of \BPL-clauses. 
\begin{enumerate}[(i)]
\item
\label{item_resol_decides_BPL}
Unrefined binary resolution with condensing (or subsumption deletion) is guaranteed to terminate on $N$, and
decides satisfiability of $N$. 
\item
The same holds for any reﬁnement of resolution with condensing (or subsumption deletion).
\end{enumerate}
\end{theorem}

Theorem~\ref{theorem_resol_decides_BPL}(\ref{item_resol_decides_BPL})~also
gives us a decision procedure for consequence finding in \BPL,
because binary resolution and subsumption is in general consequence
finding complete:

\begin{theorem}[\cite{Lee67}]
Given a set of clauses $N$, for any clause $D$ entailed by $N$,
i.e., $N \models D$,
there is a clause $C$ derivable using binary resolution from $N$ such
that $C$ subsumes $D$ (which implies $\models \forall C \mlimpl \forall
D$, where $\forall C$ denotes the formula obtained by universally
quantifying all free variables in $C$).
\end{theorem}

This means that binary resolution computes strongest consequences modulo
subsumption, which in turn implies that \BPL has the uniform
interpolation property and the approach is uniform interpolation
complete for \BPL:

\begin{theorem}\label{UI_BPL}
Binary resolution and subsumption deletion decides uniform
interpolation for \BPL and computes \BPL-uniform interpolants.
\end{theorem}

\section{The Ordered Fragment and Constant Skolemisation}
\label{section_OF_CSk}

Let $X_m = \{ x_1, \ldots, x_m\}$ denote an ordered set of variables.
A formula in the \emph{ordered fragment} is an \emph{ordered formula over $X_n$} for
some $n \geq 0$, defined inductively as follows:
\begin{enumerate}[(i)]
\item For any $n$-ary predicate symbol $P$, $P(x_1, \ldots, x_n)$ is an
ordered formula over~$X_n$.
\item
Any Boolean combination of ordered formulas over $X_n$
is an ordered formula over~$X_n$.
\item
If $\phi$ is an
ordered formula over~$X_{n+1}$, then
$\exists x_{n+1} \phi$  and   $\forall x_{n+1} \phi$
are ordered formulas over $X_n$.
\end{enumerate}
Recall the example we saw before:
\[\forall x_1\exists x_2(R(x_1,x_2)\land \forall x_3\neg S(x_1,x_2,x_3)).\]
We can easily check that it is an ordered formula over $X_0$.

We observe that the free variables in ordered formulas over~$X_n$ are exactly $x_1,\dots,x_n$, and, in particular, ordered formulas over $X_0$ are closed. We thus refer to ordered formulas over $X_0$ as \emph{ordered sentences}.

Let $\psi$ be a first-order formula in prenex normal form. 
We use $\Upsilon \psi$ to denote the result of moving universal quantifiers to the front, i.e., the quantifier prefix in $\Upsilon \psi$ is of the form $\forall^*\exists^*$.
The operator $\Upsilon$ was called the \emph{quantifier exchange operator} in \cite{OhlbachSchmidt97}.
Although, in general, $\Upsilon \psi$ is not logically equivalent to
$\psi$, we have the following result. 
\begin{theorem}[\cite{OhlbachSchmidt97}]\label{thm:sat_invariance}
Let $\psi$ be the prenex normal form of an ordered sentence. Then
$\models \Upsilon \psi$ iff $\models \psi$, and 
$\neg \Upsilon \psi$ is satisfiable iff $\neg \psi$ is satisfiable.
\end{theorem}
Notice that by pushing the negation inwards in $\neg \Upsilon \psi$ we get a $\exists^*\forall^*$-formula.
Moreover, transforming $\neg \Upsilon\psi$ into clausal form
results in a clause set in \BPL~\cite{Schmidt97d,Schmidt99a}.
Thus, satisfiability in the ordered fragment can also be decided by resolution without special refinements.

However, if we are to use SCAN or resolution to compute uniform interpolants for ordered sentences, the invariance of satisfiability in Theorem~\ref{thm:sat_invariance} is insufficient. We need to show that by transforming an ordered sentence into a set of clauses in \BPL, its logical consequences in the ordered fragment are not affected. For this purpose, we introduce \emph{constant Skolemisation} (for ordered sentences) and prove that it preserves logical consequences in the ordered fragment.

In constant Skolemisation, we introduce new \emph{constants} for all eliminated quantifiers (so no functional terms are added). The definition below is adapted from that of \emph{inner Skolemisation}~\cite{Nonnengart1996,NonnengartWeidenbach01} presented in \cite[\S3.3.5]{GabbaySchmidtSzalas08}. 
Let $\phi$ be a first-order formula in which no occurrence of subformulas of $\phi$ is of zero polarity. The \emph{constant Skolemisation of $\phi$}, written $\CSk{\phi}$, is the result of exhaustive application of the following operation.
\begin{enumerate}
    \item If $\delta$ is an occurrence of a subformula of the form $\exists x\gamma$ and is of positive polarity, pick a fresh constant $c$ and replace $\delta$ by $\gamma(x/c)$, the result of substituting $c$ for $x$ in $\gamma$.
    \item If $\delta$ is an occurrence of a subformula of the form $\forall x\gamma$ and is of negative polarity, pick a fresh constant $c$ and replace $\delta$ by $\gamma(x/c)$, the result of substituting $c$ for $x$ in $\gamma$.
\end{enumerate}
For example,
\begin{align*}
    & \CSk{\forall x\exists y\forall z P(x,y,z)} = \forall x\forall z P(x,c,z)\\
    & \CSk{\exists x\forall y (R(x,y)\lor \neg\forall z S(x,y,z))} = \forall y (R(c_1,y)\lor \neg S(c_1,y,c_2)).
\end{align*}
During the construction of $\CSk{\phi}$, if the quantifier in a subformula occurrence $\delta$ in $\phi$ is eliminated by introducing a constant~$c$, we say that $c$ is \emph{the Skolem constant associated with~$\delta$}. Thus, in the second example, $c_1$ is the Skolem constant associated with $\exists x\forall y (R(x,y)\lor \neg\forall z S(x,y,z))$, and $c_2$ is the one associated with $\forall z S(x,y,z)$.

In the rest of this section we show that logical consequences in the ordered fragment is invariant under constant Skolemsation.
\begin{lemma}\label{lemma_monotonicity}
Let $\phi$ be a first-order sentence in which no subformula occurrence is of zero polarity, and $\CSk{\phi}$ its constant Skolemisation. Then: $\models \CSk{\phi}\to\phi$.
\end{lemma}

\begin{proof}
During the construction of $\CSk{\phi}$, if an occurrence of a subformula $\exists x\gamma$ with positive polarity is replaced by $\gamma(x/c)$, we have $\models \gamma(x/c)\to \exists x\gamma$; if an occurrence of a subformula $\forall x\gamma$ with negative polarity is replaced by $\gamma(x/c)$, we have $\models\forall x\gamma \to \gamma(x/c)$. The result then follows by the Implicational Replacement Theorem (see \cite[Thm. 8.2.4]{Fitting1996}).
\end{proof}

Therefore, for any ordered sentences $\phi$ and $\psi$, if $\models\phi\to\psi$ then $\models\CSk{\phi}\to\psi$. To show that logical consequences are also preserved in the other direction, i.e., if $\models\CSk{\phi}\to\psi$ then $\models\phi\to\psi$, we need a characterisation of the expressive power of the ordered fragment. The bisimulation for the ordered fragment, or \emph{ordered bisimulation}, is developed in \cite{Bednarczyk2022} in a more general setting. In the following we offer a definition adapted for the present work.
 
Given structures $\mathfrak{A}$ and $\mathfrak{B}$ interpreting the same signature, an ordered bisimulation between $\mathfrak{A}$ and $\mathfrak{B}$ is a binary relation $Z\subseteq \bigcup_{n\geq 0}(A^n\times B^n)$ such that for any $\bar{a}\in A^n$, $\bar{b}\in B^n$ with $\bar{a} Z \bar{b}$, the following conditions hold:
\begin{enumerate}
    \item for any $n$-ary predicate $P$ in the signature, $\bar{a}\in P^\mathfrak{A}$ iff $\bar{b}\in P^\mathfrak{B}$;
    \item for any $a\in A$ there is $b\in B$ such that $\bar{a}aZ\bar{b}b$;
    \item for any $b\in B$ there is $a\in A$ such that $\bar{a}aZ\bar{b}b$.
\end{enumerate}

\begin{example}
Let $P,Q$ be unary and $R$ binary, and let $\mathfrak{A}$ be the $\{P,Q,R\}$-structure such that:
\begin{align*}
A& =\{a_1,a_2\};\\
P^{\mathfrak{A}}& =\{a_1,a_2\}, \qquad Q^{\mathfrak{A}}=\{a_2\}; \\
R^{\mathfrak{A}}& =\{(a_1,a_1),(a_1,a_2),(a_2,a_1)\}.
\intertext{%
Let $\mathfrak{B}$ be the $\{P,Q,R\}$-structure such that:}
B& =\{b_1,b_2\};\\
P^{\mathfrak{B}}& =\{b_1,b_2\}, \qquad Q^{\mathfrak{B}}=\{b_1\};\\
R^{\mathfrak{B}}& =\{(b_1,b_2),(b_2,b_1),(b_2,b_2)\}.
\end{align*}
In the following diagram, the dashed lines represent an ordered bisimulation between $\mathfrak{A}$ and~$\mathfrak{B}$, restricted to tuples of length at most $2$ ($\varepsilon$ denotes the empty tuple).
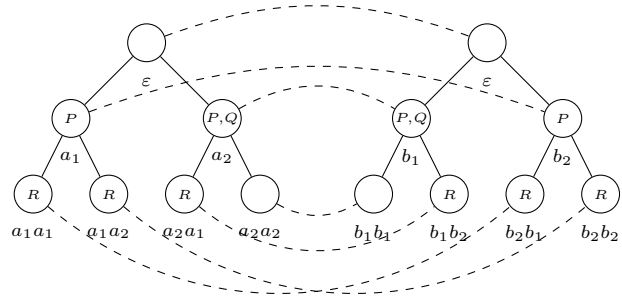
\begin{center}
\begin{tikzpicture}[minimum size=5mm,inner sep=0]
\node[circle,draw] (e) [label=below:$\scriptstyle \varepsilon$] {};
\node[circle,draw] (a) at (-1,-1) [label=below:$\scriptstyle a_1$] {$\scriptscriptstyle P$};
\node[circle,draw] (b) at (1,-1) [label=below:$\scriptstyle a_2$] {$\scriptscriptstyle P,Q$};
\node[circle,draw] (aa) at (-1.5,-2) [label=below:$\scriptstyle a_1a_1$] {$\scriptscriptstyle R$};
\node[circle,draw] (ab) at (-0.5,-2) [label=below:$\scriptstyle a_1a_2$] {$\scriptscriptstyle R$};
\node[circle,draw] (ba) at (0.5,-2) [label=below:$\scriptstyle a_2a_1$] {$\scriptscriptstyle R$};
\node[circle,draw] (bb) at (1.5,-2) [label=below:$\scriptstyle a_2a_2$] {$\scriptscriptstyle $};

\node[circle,draw] (e') at (4.5,0) [label=below:$\scriptstyle \varepsilon$] {};
\node[circle,draw] (a') at (3.5,-1) [label=below:$\scriptstyle b_1$] {$\scriptscriptstyle P,Q$};
\node[circle,draw] (b') at (5.5,-1) [label=below:$\scriptstyle b_2$] {$\scriptscriptstyle P$};
\node[circle,draw] (aa') at (3,-2) [label=below:$\scriptstyle b_1b_1$] {$\scriptscriptstyle $};
\node[circle,draw] (ab') at (4,-2) [label=below:$\scriptstyle b_1b_2$] {$\scriptscriptstyle R$};
\node[circle,draw] (ba') at (5,-2) [label=below:$\scriptstyle b_2b_1$] {$\scriptscriptstyle R$};
\node[circle,draw] (bb') at (6,-2) [label=below:$\scriptstyle b_2b_2$] {$\scriptscriptstyle R$};

\path[-] (e) edge (a);
\path[-] (e) edge (b);
\path[-] (a) edge (aa);
\path[-] (a) edge (ab);
\path[-] (b) edge (ba);
\path[-] (b) edge (bb);

\path[-] (e') edge (a');
\path[-] (e') edge (b');
\path[-] (a') edge (aa');
\path[-] (a') edge (ab');
\path[-] (b') edge (ba');
\path[-] (b') edge (bb');

\path[-] (e) edge[bend left=20,dashed] (e');
\path[-] (a) edge[bend left=20,dashed] (b');
\path[-] (b) edge[bend left,dashed] (a');
\path[-] (aa) edge[bend right=40,dashed] (ba');
\path[-] (ab) edge[bend right=40,dashed] (bb');
\path[-] (ba) edge[bend right=40,dashed] (ab');
\path[-] (bb) edge[bend right,dashed] (aa');
\end{tikzpicture}
\end{center}

\end{example}

By a pointed structure we understand a pair $(\mathfrak{A},\bar{a})$, where $\mathfrak{A}$ is a structure and $\bar{a}$ is a tuple of elements of $\mathfrak{A}$.
We say that pointed structures $(\mathfrak{A},\bar{a})$ and $(\mathfrak{B},\bar{b})$ are \emph{ordered bisimilar}, written $(\mathfrak{A},\bar{a})\sim(\mathfrak{B},\bar{b})$, if there is an ordered bisimulation between $\mathfrak{A}$ and~$\mathfrak{B}$ and, in addition, $\bar{a}Z\bar{b}$. 
We write $\mathfrak{A}\sim\mathfrak{B}$ if $(\mathfrak{A},\varepsilon)\sim(\mathfrak{B},\varepsilon)$ holds, where $\varepsilon$ is the empty tuple.

\begin{lemma}\label{lemma_bisimulation}
Let $(\mathfrak{A},\bar{a})$ and $(\mathfrak{B},\bar{b})$ be pointed structures with $|\bar{a}|=|\bar{b}|=n$. If $(\mathfrak{A},\bar{a})\sim(\mathfrak{B},\bar{b})$ then $(\mathfrak{A},\bar{a})$ and $(\mathfrak{B},\bar{b})$ agree on all ordered formulas over $X_n$, i.e., for any ordered formula $\phi$ over $X_n$, $\mathfrak{A}\models\phi[\bar{a}]$ iff $\mathfrak{B}\models\phi[\bar{b}]$.
\end{lemma}

\begin{proof}
This is simple structural induction on ordered formulas.
\end{proof}

We use the following trick, taken from \cite[pp. 168--169]{Pratt2023}, to inflate models of equality-free first-order formulas. We assume that the signature in question is relational. Let $\mathfrak{A}$ be a structure (with domain $A$) and $n$ a positive integer. Then $\mathfrak{A}'$ is defined to be the structure over the domain $A\times \{1,\dots,n\}$ such that for any $m$-ary predicate $R$, any $a_1,\dots,a_m\in A$, and any $i_1,\dots,i_m\in\{1,\dots,n\}$, 
$$(\langle a_1,i_1\rangle,\dots,\langle a_m,i_m\rangle)\in R^{\mathfrak{A}'} \text{ iff } (a_1,\dots,a_m)\in R^{\mathfrak{A}}.$$
We call $\mathfrak{A}'$ the \emph{$n$-fold Cartesian product of $\mathfrak{A}$}.
Intuitively, $\mathfrak{A}'$ is the result of duplicating each element of $\mathfrak{A}$ for $n-1$ times. 

\begin{lemma}\label{lemma_inflation}
Let $\mathfrak{A}$ be a structure, $n$ a positive integer, and $\mathfrak{A}'$ the $n$-fold Cartesian product of $\mathfrak{A}$. Then, for any $a_1,\dots,a_m\in A$, any $i_1,\dots,i_m\in\{1,\dots,n\}$, and any equality-free first-order formula $\phi(x_1,\dots,x_m)$, we have that $\mathfrak{A}\models\phi[a_1,\dots,a_m]$ iff $\mathfrak{A}'\models\phi[\langle a_1,i_1\rangle,\dots,\langle a_m,i_m\rangle]$.
\end{lemma}

\begin{proof}
By induction on the structure of equality-free formulas.
\end{proof}

Now we are ready to prove the following lemma, which establishes the preservation of logical consequences in the other direction.

\begin{lemma}\label{lemma_cons_preserv}
Let $\phi$ be an ordered sentence in which no subformula occurrence is of zero polarity, and $\CSk{\phi}$ its constant Skolemisation. Then: for any $\mathfrak{A}$ such that $\mathfrak{A}\models\phi$, there exists $\mathfrak{B}$ such that $\mathfrak{B}\models \CSk{\phi}$ and $\mathfrak{A},\mathfrak{B}$ satisfy the same ordered sentences.
\end{lemma}

\begin{proof}

Let $\mathfrak{A}$ be a structure such that $\mathfrak{A}\models\phi$. We need to construct a model $\mathfrak{B}$ of $\CSk{\phi}$, which agrees with $\mathfrak{A}$ on all ordered sentences. Let $K$ be the set of Skolem constants in $\CSk{\phi}$. Recall that each member of $K$ is associated with an occurrence of a subformula of $\phi$. Suppose $|K|=n$. We let $C=A\times\{1,\dots,n\}$, and $\mathfrak{C}$ the $n$-fold Cartesian product of $\mathfrak{A}$. By Lemma~\ref{lemma_inflation}, $\mathfrak{A}$ and $\mathfrak{C}$ satisfy the same equality-free sentences, including ordered sentences.

Note that $|C|\geq |K|$. Let $B$ be a \emph{superset} of $K$ with $|B|=|C|$. 
We now construct a sequence of functions $h_0,h_1,\dots$, where each $h_i$ is a bijection from $B^i$ to $C^i$, by the following procedure. First, let $h_0(\varepsilon)=\varepsilon$, and thus $h_0$ is defined. (Recall that $\varepsilon$ denotes the empty tuple.) When $h_i$ is defined, we define $h_{i+1}$ as follows. Suppose $\bar{b}\in B^i$, $\bar{c}\in C^i$, and $h_i(\bar{b})=\bar{c}$. We proceed with the following operations in turn:
\begin{description}
\item[Step 1:] For any $b\in K$:
\begin{enumerate}
    \item[(1)] If $b\in K$ is the Skolem constant associated with an occurrence of a subformula of $\phi$ in the form $\exists x_{i+1}\gamma$, and $\mathfrak{C}\models\exists x_{i+1} \gamma[\bar{c}]$, then pick up a fresh $c\in C$ such that $\mathfrak{C}\models\gamma[\bar{c}c]$, and let $h_{i+1}(\bar{b}b)=\bar{c}c$.
    \item[(2)] If $b\in K$ is the Skolem constant associated with an occurrence of a subformula of $\phi$ in the form $\forall x_{i+1}\gamma$, and $\mathfrak{C}\nvDash\forall x_{i+1} \gamma[\bar{c}]$, then pick up a fresh $c\in C$ such that $\mathfrak{C}\nvDash\gamma[\bar{c}c]$, and let $h_{i+1}(\bar{b}b)=\bar{c}c$.
\end{enumerate}
\item[Step 2:] For any $b\in B{\setminus}K$, pick up a fresh $c\in C$, and let $h_{i+1}(\bar{b}b)=\bar{c}c$.
\end{description}
This way we have fully defined $h_{i+1}$. Note that in (1), if $\mathfrak{C}\models\exists x_{i+1} \gamma[\bar{c}]$, there are, by Lemma~\ref{lemma_inflation}, at least $n$ distinct $c\in C$ for which $\mathfrak{C}\models\gamma[\bar{c}c]$. Thus, for each $b\in K$, there is always a fresh $c$ we can use. Similarly, in (2), if $\mathfrak{C}\nvDash\forall x_{i+1} \gamma[\bar{c}]$, there are at least $n$ distinct $c\in C$ for which $\mathfrak{C}\nvDash\gamma[\bar{c}c]$, and thus there is always a fresh $c$ we can use.

Let $h=\bigcup_{i\geq 0} h_i$. We can easily verify that $h\subseteq\bigcup_{i\geq 0}(B^i\times C^i)$. Also, we define the structure $\mathfrak{B}$ over the domain $B$ such that:
\begin{align*}
R^{\mathfrak{B}} & =\{\bar{b}\in B^m\mid h(\bar{b})\in R^{\mathfrak{C}}\} \quad \text{for each $m$-ary predicate $R$;}\\
b^\mathfrak{C} & =b \quad  \text{for each $b\in K$.}
\end{align*}
It is routine to check that $h$ is an ordered bisimulation between $\mathfrak{B}$ and $\mathfrak{C}$. Since $h(\varepsilon)=\varepsilon$, we have $\mathfrak{B}\sim\mathfrak{C}$. Then, we have by Lemma~\ref{lemma_bisimulation} that $\mathfrak{B}$ and $\mathfrak{C}$ satisfy the same ordered sentences, and thus $\mathfrak{A}$ and $\mathfrak{B}$ satisfy the same ordered sentences. In particular, we have $\mathfrak{B}\models\phi$.
Also, the construction of $\mathfrak{B}$ ensures that:
\begin{description}
    \item[C1] if $b\in K$ is associated with $\exists x_{i+1}\gamma$, we have $\mathfrak{B}\models\exists x_{i+1}\gamma\to\gamma(x_{i+1}/b)$;
    \item[C2] if $b\in K$ is associated with $\forall x_{i+1}\gamma$, we have $\mathfrak{B}\models\gamma(x_{i+1}/b)\to\forall x_{i+1}\gamma$.
\end{description}
To see C1, suppose that $b\in K$ is associated with $\exists x_{i+1}\gamma$, and that $\mathfrak{B}\models\exists x_{i+1}\gamma[\bar{b}]$ for some $\bar{b}\in B^i$. Since $(\mathfrak{B},\bar{b})\sim(\mathfrak{C},h(\bar{b}))$, we have $\mathfrak{C}\models\exists x_{i+1}\gamma[h(\bar{b})]$. By the construction of $h_{i+1}$, $h(\bar{b}b)=h(\bar{c})c$ for some $c\in C$ such that $\mathfrak{C}\models\gamma[h(\bar{b})c]$; then $(\mathfrak{B},\bar{b}b)\sim(\mathfrak{C},h(\bar{b})c)$, and thus $\mathfrak{B}\models\gamma[\bar{b}b]$. Since $b$ is a constant denoting itself in $\mathfrak{B}$, and $\gamma$ is an ordered formula over $X_{i+1}$, we get $\mathfrak{B}\models\gamma (x_{i+1}/b)[\bar{b}]$ from $\mathfrak{B}\models\gamma[\bar{b}b]$. Therefore, $\mathfrak{B}\models\exists x_{i+1}\gamma\to\gamma(x_{i+1}/b)$. C2 can be proved similarly.

Finally, by C1, C2, and the Implicational Replacement Theorem (see \cite[Thm. 8.2.4]{Fitting1996}), we get $\mathfrak{B}\models\phi\to\CSk{\phi}$, and then $\mathfrak{B}\models \CSk{\phi}$.
This completes the proof.
\end{proof}

\begin{theorem}
Let $\phi$ be an ordered sentence in which no subformula occurrence is of zero polarity, and $\CSk{\phi}$ its constant Skolemisation. Then: for any ordered sentence $\psi$, we have that $\models \CSk{\phi}\to\psi$ iff $\models\phi\to\psi$.
\end{theorem}

\begin{proof}
The `if' direction follows from Lemma~\ref{lemma_monotonicity}. For the `only-if' direction, suppose ${\models \CSk{\phi}\to\psi}$ and $\mathfrak{A}\models\phi$. Then, by Lemma~\ref{lemma_cons_preserv}, there exists $\mathfrak{B}$ such that $\mathfrak{B}\models \CSk{\phi}$ and $\mathfrak{A},\mathfrak{B}$ satisfy the same ordered sentences. Since $\models \CSk{\phi}\to\psi$, we have $\mathfrak{B}\models \psi$ and then $\mathfrak{A}\models \psi$. Thus, $\models\phi\to\psi$.
\end{proof}

Finally, since the transformation to clausal form (after constant Skolemisation) involves only logical equivalences, we conclude that, for any ordered sentence $\phi$, the clausal form of $\CSk{\phi}$ has the same logical consequences in the ordered fragment as $\phi$.

\section{The SCAN Algorithm}
\label{section_SCAN_algorithm}

The input to SCAN is a set $N$ of first-order clauses and a list of
predicate variables~$\overline X$. The aim is to eliminate these
symbols from $N$ producing a set $M$ such that $$\models \exists X N \mlequiv M.$$

When SCAN is applied to a first-order formula $\varphi$, the formula is first
transformed into clausal form as implemented in the Otter theorem
prover.
This involves four steps:
(i) transforming the formula into prenex normal form $\overline{Q x} \varphi'$, which
moves the quantifiers to the outside of the formula so that $\varphi'$ is
free of quantifiers.
(ii) Skolemisation: the replacement of existentially quantified variables by Skolem
terms and dropping the existential quantifiers, which results in a formula $\forall \overline y \varphi''$ where $\varphi''$ is quantifier-free and $\overline y \subseteq \overline x$, 
In the next step, (iii), $\varphi''$
is transformed into conjunctive normal form.
The clausal form $N = \Cls(\varphi)$ is then obtained by (iv) writing $\varphi''$ as a set of clauses. 

By $\overline s$ we denote a sequence of terms $(s_1, \ldots, s_n)$.
$L(\overline s)$ denotes a literal with argument sequence $\overline s$
where $n$ is the arity of the predicate symbol of $L$ and the length of
$\overline s$.
Assuming
the lengths of $\overline s$ and $\overline t$ are $n$,
$\overline s \not\approx \overline t$ is shorthand notation for 
$s_1 \not\approx t_1 \mlor \ldots \mlor s_n \not\approx t_n$.

The idea of the SCAN algorithm is to generate sufficiently many logical
consequences of $N$ and eventually only keeping those non-redundant
clauses in which no symbols from $\overline X$ occur.
This is realised by a saturation process which interleaves (i)~inference
steps, (ii)~normalisation and redundancy elimination steps, and 
(iii)~purifying deletion steps.
These are based by the inference system $\cresol$ defined by the rules given
Tables~\ref{table_scan_inference_rules} and~\ref{table_cresol_calculus}.

\begin{table}[tb]
\caption{The inference rules of $\mathsf{C}$.}
\label{table_scan_inference_rules}
\begin{framed}
\begin{tabular}{@{}p{1\linewidth}@{}}
\textbf{$\mathsf{C}$-Resolution}\qquad
\AxiomC{$C\lor X(\overline{s})$}
\AxiomC{$D\lor \neg X(\overline{t})$}
\BinaryInfC{$C\lor D\lor \overline{s}\not\approx\overline{t}$}
\DisplayProof
\end{tabular}
provided $X$ is predicate variable, the two premises have no
individual variables in common, and the premises are distinct (to prevent
self-resolution).
\bigskip
\par
\begin{tabular}{@{}p{1\linewidth}@{}}
\textbf{$\mathsf{C}$-Factoring}\qquad
\AxiomC{$C\lor (\neg)X(\overline{s})\lor(\neg)X(\overline{t})$}
\UnaryInfC{$C\lor (\neg)X(\overline{s})\lor \overline{s}\not\approx\overline{t}$}
\DisplayProof
\end{tabular}
provided $X$ is a predicate variable.
\end{framed}
\end{table}

\begin{table}[tb]
\caption{The $\mathsf{C}$-resolution system.}
\label{table_cresol_calculus}
\begin{framed}
\begin{tabular}{@{}p{1\linewidth}@{}}
\textbf{Inference}\qquad
\AxiomC{$N$}
\UnaryInfC{$N\cup\{\mathrm{norm}(C)\}$}
\DisplayProof
\end{tabular}
where $C$ is a $\mathsf{C}$-resolvent or $\mathsf{C}$-factor of clauses in $N$,
and $C \Rightarrow_{\mathrm{CEl}}^*\mathrm{norm}(C)$. 

\bigskip
\begin{tabular}{@{}p{1\linewidth}@{}}
\textbf{Redundancy elimination}\qquad
\AxiomC{$N$}
\UnaryInfC{$M$}
\DisplayProof
\end{tabular}
if $M$ is obtained from $N$ by redundancy elimination.

\bigskip
\begin{tabular}{@{}p{1\linewidth}@{}}
\textbf{Purified clause deletion}\qquad
\AxiomC{$N\cup\{C\lor (\neg)X(\overline{s})\}$}
\UnaryInfC{$N$}
\DisplayProof
\end{tabular}
provided $X$ is a predicate variable and $C\lor (\neg) X(\overline s)$ is \emph{purified in $N$ with respect
to the $X$-literal} $(\neg) X(\overline s)$, i.e.,
no non-redundant inferences are possible with the clause
$C\lor (\neg) X(\overline s)$ upon the literal $(\neg) X(\overline s)$ 
and clauses in $N$.

\bigskip
\begin{tabular}{@{}p{1\linewidth}@{}}
\textbf{Extended purity deletion$^{+(-)}$}\qquad
\AxiomC{$N\cup M$}
\UnaryInfC{$N$}
\DisplayProof
\end{tabular}
if $N$ and $M$ are disjoint, $N$ is free of predicate variable $X$,
and every clause in~$M$ contains $X$ positively (negatively).
\end{framed}
\end{table}

Table~\ref{table_scan_inference_rules} specifies the inference rules
of SCAN performing all inferences upon $\overline X$-literals.
The \emph{\cresol-resolution} and \emph{\cresol-factoring} rules are variations of the
standard resolution and factoring rules in which constraints are used
instead of unification.
It is convenient to think of a \cresol-resolvent as a conditional
formula saying, if each of the corresponding terms are equal (i.e., $s_i
\foeq t_i$), then $C\mlor D$ holds.
Similarly, for \cresol-factoring.

We assume that \emph{clauses are eagerly normalised using:}
\begin{description}
    \item[\textbf{Constraint elimination:}] 
\quad $C\lor s\not\approx t\ \Rightarrow_{\mathrm{CEl}}  \ C\sigma$\\
if $\sigma$ is a~most general unifier of $s\approx t$,
i.e., $s\sigma\approx t\sigma$
and $\sigma$ is the most general substitution with this property.
\end{description}

The constraint elimination rule is useful for simplifying
clauses, reducing their size and making them more digestable by removing as many constraints
as possible, because otherwise clauses grow in width very rapidly.
Eager constraint elimination improves the performance and success
rate because other simplification possibilities such as those in
Table~\ref{table_simplification_rules} can be
recognised and applied more efficiently during the saturation
process~\cite{GabbayOhlbach92a,Engel96}.
Another purpose of constraint elimination is that it performs
reflexivity resolution required for equational reasoning in the
presence of inequality literals. 
Due to the absence of positive equational literals no further rules are
needed for equality inference.
In addition, the characterisation of the search space for the clausal form of the ordered
fragment in the next section is easier if it is assumed that clauses
contain only constraints between distinct constants.

The \cresol-resolution system can be freely enhanced by any
equivalence preserving \emph{redundancy elimination} rules to improve
efficiency while preserving correctness.
Without redundancy elimination the success rates are in general lower. Table~\ref{table_simplification_rules} lists some common redundancy
elimination rules which are highly effective in saturation-based
provers.

The \emph{purified clause deletion} rule deletes a clause with a designated
$X$-literal if no inference upon this literal produces clauses
not already implied by other clauses in the current set.

There are two \emph{extended purity deletion} rules: The positive version
deletes all clauses containing $X$, if these clauses all contain $X$
positively (i.e., unnegated, with positive polarity). 
The intuition is that then instantiating~$X$ with $\top$ makes these
clauses true.
Dually, if all clauses containing $X$ contain $X$ negatively, then these
clauses can be made true by substituting $\bot$ for $X$.
(Note that standard purity deletion is an instance of extended purity
deletion.)

\begin{theorem} \
\begin{enumerate}[(i)]
\item
Each inference, constraint elimination and simplification rule $N/M$
preserves logical equivalence:
$N \mlequiv M$, which implies
$\exists \overline X N \mlequiv \exists \overline X M$.
\item
Every other rule in the system $\cresol$ (i.e., purified clause
deleion and extended purity deletion) preserve eqivalence up to
existential quantification of the predicate variables in $\overline X$:
$\exists \overline X  N \mlequiv \exists \overline X M$.
\end{enumerate}
\end{theorem}

Let $N$ be the input clause set and suppose the aim is to eliminate the
variables $\overline X$.
The idea of the derivation process in SCAN is the following:
First, processing of the input clauses is done, which may include
$\cresol$-factoring with eager constraint elimination, subsumption deletion, tautology deletion and
creation of appropriate data structures.

Then, the main derivation loop is performed until all symbols in
$\overline X$ have been eliminated. Each iteration consists of the
following steps:
\begin{enumerate}[(a)]
\item
A clause $K$ and an $X$-literal $L$ in the clause is picked for \emph{purification}:
All $\cresol$-inferences are performed with $K$ upon $L$ and the remaining
clauses as well as all obtained clauses in this sub-loop.
Each inferred clause is checked if it is subsumed by a present clause or
subsumes some present clauses.
If no more $\cresol$-inferences are possible with $K$ upon $L$ then
$K$ is \emph{purified} and is deleted.
\item
Extended purity deletion (positive and negative) is applied.
\item
Other redundancy elimination steps may be performed at any time.
\end{enumerate}
SCAN finally attempts to express the obtained clause set $M$ as a first-order
formula by reversing Skolemisation. 

There are various choices that can be made in the algorithm and
different derivations are obtained depending on the choices. SCAN
may or may not terminate.
If SCAN terminates normally then the remaining clauses do not contain
any of the predicate variables in $\overline X$.

The choices also influence if back-translation to a first-order formula is possible.
The problem is that derived clauses can contain different complex Skolem
terms requiring Henkin quantifiers which cannot be
captured by a linear sequence of quantifiers~\cite{GabbayOhlbach92a}.

\begin{theorem}
SCAN is correct:
$\exists \overline X  N \mlequiv M$, where $M$ is the set of persisting clauses when
SCAN terminates.
\end{theorem}

As already said, the SCAN algorithm is not guaranteed to terminate; in fact, in general, there cannot
be an algorithm that necessarily succeeds, even if a finite
solution exists,
since the SOQE problem is not even semi-decidable.

\section{The Prefix Stable Clausal Class \BPLplus}
\label{section_BPLplus}

The search space of SCAN on the clausal form of ordered sentences can be
characterized
as an adaptation of the clausal class \BPL in which constants occur in fixed
argument positions in atoms and clauses possibly include
negative constraints, i.e., the inequality literals.
We name the class \BPLplus (basic path clauses with inequality
constraints and fixed constant positions).

Formally, the signature of \BPLplus is any signature $\Sigma$ defined in Section~\ref{section_BPL_resolution} also including equality $\foeq$. 
We assume that each constant $a$ is associated with an index~$i$ in $\{1, \ldots, n \}$, where $n$ is the maximal arity of any predicate symbol in the given ordered sentence.

\BPLplus is the class of all $\exists^*\forall^*$~clauses 
in which additionally the following conditions all hold.
For any two non-constraint literals $(\neg)P(\underline u)$ and $(\neg)Q(\underline v)$ in $C$, both the following hold:
\begin{description}
\item[\textbf{L1}]
if $u_i$ and $v_j$ are the same variable then $i=j$, and
\item[\textbf{L2}] 
each pair $u_k$ and $v_k$ preceding $u_i$ and $v_i$ are identical,
unless $u_k$ and $v_k$ are both (possibly distinct) constants.
\item[\textbf{L3}]
For any constant~$a$, if the index of $a$ is $i$ ($1\leq i \leq n$) then $a$ occurs only at
argument position $i$ in non-constraint literals of any clause.
\item[\textbf{L4}]
In each constraint $u \foneq v$ of clause~$C$, $u$ and $v$ are
\emph{position compatible}, i.e., 
if $u$ is a variable occuring at argument position $i$ in a non-constraint literal of $C$ 
or $u$ is a constant with index $i$ 
then either $v$ is a variable occurring at argument position $i$ in a non-constraint literal of $C$,
$v$ is singular (does not occur anywhere else in~$C$) or 
$v$ is a constant with index $i$.
We say that the \emph{position associated with $u \foneq v$} is $i$.
\end{description}

\emph{L1 and L2 say that clauses are prefix-stable for variables when
constants are disregarded.
\emph{L3} says that argument positions of constants are unique.
\emph{L4} says that constraints are position compatible.}
Observe that \emph{L3} holds for both sets of clauses in \BPLplus and
individual clauses, whereas \emph{L1}, \emph{L2} and \emph{L4} are
clause specific.
When we say a set of sequences of variables and constants is
\emph{prefix stable} (for variables) \emph{modulo disregard for
constants} we mean that the set satisfies \emph{L1} and \emph{L2}.

For example, the following clauses satisfy \emph{L1}--\emph{L4}
\begin{align*}
Q(a,y) \mlor P(a,y), \quad
Q(a,y) \mlor P(b,y), \quad
Q(a,y) \mlor P(b,y) \mlor a\foneq b \mlor c\foneq d
\end{align*}
while 
$Q(a,y) \mlor P(x,y)$
does not, since \emph{L2} does not hold.

Suppose $\varphi$ is a sentence in the ordered fragment (without
equality).
Let $N$ be a clausal form of the constant Skolemisation of~$\varphi$.
It is not difficult to see that clauses in $N$ contain only
variables and constants, are prefix-stable
as defined at the beginning of Section~\ref{section_BPL_resolution}
and the uniqueness of argument
position of constants (\emph{L3}) is satisfied.
Thus $N$ satisfies both the more general property of prefix
stability for variables (\emph{T1} and \emph{T2}) and 
prefix stability for variables modulo disregard for constants
(\emph{L1} and \emph{L2}).
Since $N$ is free of equality and it is during the constraint
resolution or constraint factoring steps that
constraints (inequality literals) are introduced, condition~\emph{L4} vacuously holds. Therefore:

\begin{theorem}
\label{theorem_input_clauses_in_BPLplus}
Let $\varphi$ be an ordered sentence and let $N = \Cls(\CSk{\varphi})$.
Then, $N \psubset \BPL$ and $N \psubset \BPLplus$.
\end{theorem}

\begin{proof}
By the results of \cite{Schmidt97d,Schmidt99a}, since the clausal form
obtained with constant Skolemisation is the same as the clausal form of
the $n$-ary optimised functional translation~\cite{SchmidtHustadt13}
of the \ml{KD}-modal formula corresponding to $\varphi$.
Hence $N \psubset \BPL$ and it can be checked that $N$
satisfies conditions \emph{L1}--\emph{L4} and belongs to \BPLplus.
\end{proof}

One might have expected that the inference rules of the system
$\cresol$ preserve
prefix-stability for variables, but the following example shows the
constraint resolvent of two \BPL-clauses after normalisation may not
be prefix-stable for variables, where the aim is to eliminate $X$.
\begin{align*}
1.\ & \hphantom{\neg} X(a,y) \mlor P(a,y) \\
2.\ & \neg X(b,z) \mlor Q(b,z)\\ 
\text{\sout{3.}}\ & \hphantom{\neg} P(a,y) \mlor Q(b,z) \mlor a \not\approx b \mlor y \not\approx z
&& \text{(1,2)}\\ 
3'.\ & 
\hphantom{\neg} P(a,y) \mlor Q(b,y) \mlor a \not\approx b 
&& \text{(norm 3)}
\end{align*}
Clause $3'$. does not satisfy \emph{T2} and is therefore not prefix stable for variables. We note however all clauses satisfy \emph{L1}--\emph{L4}.
The reason is that, in the inference rules of~$\cresol$, 
unification has been replaced
by constraints and inferences are performed even when the literals
resolved or factored upon are not unifiable.
Unification is sufficient for refutation finding
(and uniform interpolation where the target language does not include
equality) 
but not for SOQE where
we are interested in solutions expressed in first-order logic with equality.
The SOQE solution of eliminating~$X$ from clauses 1.\
and 2.\ is clause $3'$.\ which says that if 
$a \approx b \mlimpl (P(a,y) \mlor Q(b,y))$,
whereas the uniform interpolant in first-order logic
without equality is $\top$.

\begin{theorem}
\label{theorem_BPLplus_closed_under_SCAN_rules}
\BPLplus is closed under the rules in the system $\cresol$
in which the redundancy elimination rules are those listed in
Table~\ref{table_simplification_rules}, i.e., tautology deletion,
subsumption deletion and condensation.
\end{theorem}

To prove this theorem, we need
to show that each inference step in SCAN preserves prefix-stability for
variables when constants are disregarded (\emph{L1} and \emph{L2}),
constants do not change positions in non-constraint literals (\emph{L3}) and are
position compatible in constraints (\emph{L4}).
We consider the inference rules constraint resolution, constraint
factoring and eager constraint elimination in turn.

\subparagraph{Constraint resolution}
Let $C = X(\overline u) \mlor C' \mlor C'^{\not\approx}$ and
$D = \neg X(\overline v) \mlor D' \mlor D'^{\not\approx}$
be variable disjoint copies of clauses in \BPLplus and assume they are both
normalised, i.e., contraint elimination is not applicable.
Suppose constraint resolution is applicable to $C$ and $D$.
This means $\overline u$ and $\overline v$ have the same length. 
Assume $C'^{\not\approx}$ and~$D'^{\not\approx}$ are the respective
subclauses containing the constraints of $C$ and $D$.

A resolvent of $C$ and $D$ is
$$C' \mlor C'^{\not\approx} \mlor D' \mlor D'^{\not\approx} \mlor
\overline u \not\approx \overline v.$$
The part $C' \mlor C'^{\not\approx} \mlor D' \mlor D'^{\not\approx}$ 
is the disjoint union of subclauses of $C$ and $D$, which implies
\emph{L1}, \emph{L2}, \emph{L3} and \emph{L4} are preserved.
The positions of the terms in each of the new constraints $u_i \not\approx v_i$
are their respective argument positions in the literals resolved upon,
which are the same.
We can conclude the resolvent satisfies \emph{L4} and belongs to~\BPLplus.

\subparagraph{Constraint factoring}
Similarly, let $C = (\neg)X(\overline u) \mlor (\neg)X(\overline v) \mlor C'$
be any normalised clause in \BPLplus.
Suppose constraint factoring is applicable and produces factor 
$$(\neg)X(\overline u) \mlor C' \mlor \overline u \not\approx \overline v.$$
The part $X(\overline u) \mlor C'$
is unchanged and satisfies \emph{L1}--\emph{L4}.
We note that $u_i$ and $v_i$ of each constraint $u_i \foneq v_i$ occur in the same argument
position in the literals $X(\overline u)$ and $X(\overline v)$.
This implies \emph{L4} holds and the factor is a \BPLplus-clause.

\subparagraph{Constraint elimination}
Constraint elimination applies to a clause 
$C \mlor \overline u \foneq \overline v$, 
where~$C$ is normalised and therefore does not contain any
constraints involving variables.
Suppose $\overline u \foneq \overline v$ is the disjunction
\begin{equation}
\label{constraint_disjunction}
u_1 \not\approx v_1 \mlor \ldots \mlor u_n \not\approx v_n.
\end{equation}
Because $\overline u$ and $\overline v$ are variable disjoint and
linear (since they satisfy \emph{L1} and \emph{L2}), 
the bindings $\sigma$ produced for unifiable pairs $u_k \not\approx v_k$, i.e., 
$\sigma = \{ u_k \mapsto v_k \}$ or $\sigma = \{ v_k \mapsto u_k \}$,
are independent in the following sense:
No variable changed in a binding occurs on the right of another binding
of the bindings produced during $\Rightarrow_{\mathrm{CEl}}^*$, i.e.,
exhaustive application of constraint elimination to the clause.
This implies that the bindings and individual contraint elimination
steps can be applied in any order to the
subclause~(\ref{constraint_disjunction})
leading to the same normalised clause.
Without loss of generality we assume eager constraint elimination
is performed in the order 
from left to right in~(\ref{constraint_disjunction}).

Observe that constraints of the form $a \foneq a$ or $x \foneq x$
(the latter is a possibility in a factor) are simply dropped
during constraint elimination (in this case the binding is~$\sigma=\emptyset$).
$u_k \not\approx v_k$ is not unifiable when the terms are distinct
constants and $u_k$ and~$v_k$ are left untouched and the constraint
remains in the result of applying $\Rightarrow_{\mathrm{CEl}}^*$.

The following lemma shows that a binding $\sigma$ during constraint elimination from
left to right preserves \emph{L1} and \emph{L2}.
It also trivially preserves \emph{L3} and \emph{L4}.

Let $A$ be a set of argument sequences in the signature of \BPLplus.
We say two argument sequences $s$ and $t$ (of equal length) are
\emph{$k$-equal} if $s$ and $t$ are equal except possibly at
position $k$, that is, for every position $i \not= k$, $s|_i = t|_i$.

Recall, the preﬁx of $u_i$ in $\overline u$ is the empty sequence $\varepsilon$ if $i = 1$,
and $(u_1, \ldots, u_{i-1})$, otherwise.
We define the \emph{variable preﬁx} of $u_i$ in $\overline u$ as 
the sequence $(u_{j_1}, \ldots, u_{j_{k-1}})$ of all variables in
$\overline u$, where
$j_1 < \ldots < j_{k-1} < i$ and all the other terms in
the prefix of $u_i$ are constants.
If none of the $u_1, \ldots, u_{i-1}$ are variables the variable
prefix is the empty sequence $\varepsilon$.

\begin{lemma}
\label{L1_L2_set_binding}
Let $(u_1 \ldots u_m)$ and $(v_1 \ldots v_n)$ be two argument sequences
in $A$ with the same variable prefixes of $u_{j_k}$ and $v_{j_k}$ with
$u_{j_k} \not= v_{j_k}$, i.e.,
for some $k > 0$, the variable positions $j_1, \ldots, j_{k}$ in 
$\overline u$ and $\overline v$ where $j_1 < \ldots < j_{k}$ we have
that
\begin{equation}
\label{L1_L2_set_binding_condition}
u_{j_1} = v_{j_1}, \ldots, u_{j_{k-1}} = v_{j_{k-1}} \quad \text{but} \quad  u_{j_k} \not= v_{j_k}.
\end{equation}
Furthermore, assume that $u_{j_k}$ is a variable.
Let $\sigma$ be the substitution $\{ u_{j_k} \mapsto v_{j_k} \}$.
Then $A\sigma$ satisfies L1 and L2, provided $A$ does.
\end{lemma}
To prove this lemma, we 
consider two arbitrary sequences in $A\sigma$.
They are of the form $s \sigma$ and $t \sigma$ with $s$ and $t$ some
sequences in $A$.
For $s$ and $t$ conditions \emph{L1} and \emph{L2} hold, and
we want to show they also hold for the sequences $s \sigma$ and $t \sigma$.

\begin{lemma}
The sequences $s \sigma$ and $s$ are $j_k$-equal and differ only when $s|_{j_k} = u_{j_k}$.
\end{lemma}
\begin{proof}
$\sigma$ affects only the variable $u_{j_k}$ and in any sequence of $A$, $u_{j_k}$
occurs only at position~${j_k}$ else condition \emph{L1} is violated.
Hence, if $u_{j_k}$ occurs in $s$ then (i)~$s|_{j_k} = u_{j_k}$, 
(ii)~for any $l \not= {j_k}$, $s|_l \not= u_{j_k}$, and 
(iii)~$s \sigma|_{j_k} = v_{j_k} \not= s|_{j_k}$.
\end{proof}
We continue the proof of Lemma~\ref{L1_L2_set_binding}.
The lemma just proved is true for $t \sigma$ and $t$, as well.
If neither $s$ nor $t$ contain the variable $u_{j_k}$ then the substitution
$\sigma$ does not affect $s$ and $t$.
Then $s \sigma = s$ and $t \sigma = t$.
In this case $s \sigma$ and $t \sigma$ trivially satisfy \emph{L1} and~\emph{L2}
(since $s$ and $t$ do).

Without loss of generality let us assume that $s|_{j_k} = u_{j_k}$.
Then $s|_{j_k} \sigma = v_{j_k}$.
We distinguish two cases:
\begin{enumerate}[(i)]
\item
$t|_{j_k} \not= u_{j_k}$ and $t|_{j_k} \not= v_{j_k}$.
$\sigma$ leaves $t$ unchanged so that $t \sigma = t$.
Suppose $s \sigma|_i = t \sigma|_j$ is a variable.
Then $s \sigma|_i = t \sigma|_j = t|_j$.
Also, $j \not= {j_k}$ and $i \not= {j_k}$, since otherwise $t|_j =
v_{j_k}$
which contradicts our assumption.
This implies $s|_i = s \sigma|_i =  t \sigma|_j = t|_j$.
By \emph{L1} which holds for $s$ and $t$ we get that $i = j$.
By \emph{L2} for any $l < i=j$ we have $s|_l = t|_l$ unless both $s|_l$ and $t|_l$
are (possibly distinct) constants.
Suppose $s|_l$ and $t|_l$ are not both constants.
Then, $i = j < {j_k}$ since otherwise, if $i = j = {j_k}$ then $s|_i = u_{j_k}
\not= t|_i$ which is a contradiction, or if $i = j > {j_k}$ then since
$s|_{j_k} \not= t|_{j_k}$ by assumption, $s$ and~$t$ contradict \emph{L1} and \emph{L2}.
Consequently by the Lemma $s \sigma|_l = s|_l = t|_l = t \sigma|_l$.
Therefore, conditions \emph{L1} and \emph{L2} are true for case (i).

\item
$t|_{j_k} = u_{j_k}$ or $t|_{j_k} = v_{j_k}$.
Then $t \sigma|_{j_k} = v_{j_k} = s \sigma|_{j_k}$.
Suppose $s \sigma|_i = t \sigma|_j$ is a variable.
\begin{enumerate}[a.]
\item
If $i = {j_k}$ then $s \sigma|_i = v_{j_k} = t \sigma|_j$.
Then, either $t|_j = u_{j_k}$ or $t|_j = v_{j_k}$.
In either case, it follows that $j = {j_k}$ and hence $i=j$.
\item
If $j = {j_k}$ then by a similar argument $i= j$.
\item
If $i \not= {j_k}$ and $j \not= {j_k}$ then the Lemma implies $s \sigma|_i =
s|_i$ and $t \sigma|_j = t|_j$.
Since $s \sigma|_i = t \sigma|_j$ we have $s|_i = t|_j$ and it follows
by \emph{L1} that $i = j$.
\end{enumerate}
Therefore, $s \sigma$ and $t \sigma$ satisfy \emph{L1}.

Let $l < i=j$ be arbitrary and suppose $s|_l$ and $t|_l$ are variables
(by \emph{L2} it cannot be that one is a variable and the other is a constant).
By \emph{L2} we have that $s|_l = t|_l$.
\begin{enumerate}[a.]
\item
Consider the case that $l \not= {j_k}$.
By the Lemma $s \sigma|_l = s|_l$ and $t \sigma|_l = t|_l$.
Since $s|_l$ and $t|_l$ coincide, we conclude $s \sigma|_l =
t \sigma|_l$.
(Note that if $t|_{j_k} = v_{j_k}$ then $s|_{j_k} = u_{j_k} \not= t|_{j_k}$ and
consequently $i=j < {j_k}$.)
\item
For $l = {j_k}$:
$s \sigma|_l = v_{j_k} = t \sigma|_l$ by assumption.
\end{enumerate}
\end{enumerate}
This completes the proof of Lemma~\ref{L1_L2_set_binding}.

We can conclude that the normalisation obtained by eager constraint
elimination of any resolvent or factor in $\cresol$ preserves
\emph{L1}--\emph{L4} and belongs to \BPLplus.

Trivially, the class \BPLplus is closed under clause deletion rules,
including purified clause deletion, extended purity deletion,
tautology deletion and subsumption deletion. To see that the
condensation of a \BPLplus-clause is a \BPLplus-clause 
notice that \emph{the condensation of a clause is a factor that subsumes it},
and use Lemma~\ref{L1_L2_set_binding} 
to show the following:
\begin{quote}
\emph{Let $A$ be a set of argument sequences in the language of
\BPLplus-clause and assume $\sigma$ be an
idempotent unifier of two sequences $\overline u$ and $\overline v$ in $A$.
Then, $A\sigma$ satisfies properties L1 and~L2 if $A$ does.}
\end{quote}
This allows us to show standard factoring preserves properties \emph{L1} and
\emph{L2}. Preservation of \emph{L3} and \emph{L4} is not difficult to see.

This completes the proof of
Theorem~\ref{theorem_BPLplus_closed_under_SCAN_rules}.

\bigskip
We now argue that SCAN is a decision procedure for eliminating predicate
variables $\overline X$ from any finite set $N$ of \BPLplus clauses.
This requires us to show that any $\cresol$-saturated set $N^\infty$
of $N$ is finite, which follows if the depth and width
of any derived clauses are finitely bounded.

Trivially, the clause (term) depth is bounded as \BPLplus clauses do not
contain complex terms, only variables and constants.

The existence of a finite bound of the
size of derived clauses, i.e., the number of literals in derived clauses
does not grow indefinitely, is tied with the maximum number of
variables which can occur in any variable indecomposable (sub)clause.

Let $A$ be a set of sequences (or a clause) in \BPLplus.
The \emph{variable partition} of $A$ is the finest partition of $A$
into disjoint subsets of sequences (subclauses) which do not share common variables.
The blocks in the variable partition are said to be
\emph{variable indecomposable} or \emph{split}.
It suffices to show there is a bound on the number of variables
occurring in any variable indecomposable and condensed $A$.
Because if there is, then there can only be a limited number of
non-variant variable indecomposable subsets in \BPLplus.
In particular, there can only be a limited number of non-variant variable
indecomposable subsets of any condensed $A$.
This implies that any condensed $A$ in \BPLplus is bounded in size.
Consequently, \BPLplus is bounded in size.

Now, w.l.o.g.\ we can assume that all clauses are normalised
and all constraints are between constants.
We can further limit our attention to the non-constraint literals
in clauses.
It is not difficult to see that the maximum number of variables
occurring in any variable indecomposable \BPLplus-clause is the same
as the variable bound in any \BPL-clause (or any \BPL-clause satisfying
the fixed constant position condition).
Variable indecomposable normalised \BPLplus-clauses have in general
more literals than \BPL-clauses satisfying the fixed constant position
condition, because of the presence of the
constraints but also because of condition \emph{L2} 
which allows combinatorially many more literals in a clause involving the constants in the input set.
Since there are only finitely many constants and because of the
finite bound on variables the size of \BPLplus-clauses is finitely bounded.

\begin{theorem}
SCAN is a decision procedure for eliminating predicate
variables $\overline X$ from any finite set $N$ of \BPLplus clauses,
i.e., it decides SOQE for \BPLplus and computes
an answer in first-order logic with equality.
\end{theorem}

This implies that SCAN is SOQE-complete for \BPLplus (and the clausal form of the ordered
fragment resulting from constant Skolemisation).

\section{Extracting Uniform Interpolants in Basic Path Logic} 
\label{section_UI}
In this section we offer a construction of uniform interpolants (in equality-free first-order logic) using \BPLplus clauses computed by SCAN.
We have shown that, for any finite set of \BPLplus clauses (and any finite list of predicate variables), SCAN always terminates with a finite set of \BPLplus clauses, and thus we have a solution to SOQE in this case. 
However, since \BPLplus clauses in general are not expressible in \BPL and the ordered fragment, both of which are equality-free, some extra work is needed in order to extract uniform interpolants from \BPLplus clauses.

As constraint elimination is applied eagerly, all inequalities occurring in the output set of clauses are between (distinct) constants, e.g. $a\foneq b$. Fortunately, replacing these inequalities with $\top$, or equivalently, deleting all clauses containing such inequalities does not cause a loss of logical consequences in equality-free first-order logic. The correctness of this deletion is discussed below.
Again, we make use of the fact that models of equality-free formulas can be inflated.
\begin{lemma}
\label{lemma_constat_inequalities}
Let $\mathfrak{A}$ be a structure (with domain $A$) interpreting a relational signature with constants. There is a structure $\mathfrak{B}$ (interpreting the same signature) such that: (i) no two constants denote the same element in $\mathfrak{B}$ and (ii) $\mathfrak{A}$ and $\mathfrak{B}$ satisfy the same equality-free first-order sentences.
\end{lemma}

\begin{proof}

Let $C$ be the set of constants in the signature, and $B=A\times C$. Then we define the structure $\mathfrak{B}$ over the domain $B$ such that:
\begin{itemize}
\item[] For any $m$-ary predicate $R$, any $a_1,\dots,a_m\in A$, and any $c_1,\dots,c_m\in C$, 
$$(\langle a_1,c_1\rangle,\dots,\langle a_m,c_m\rangle)\in R^{\mathfrak{B}} \text{ iff } (a_1,\dots,a_m)\in R^{\mathfrak{A}}.$$
\item[] For any constant $c\in C$, $c^\mathfrak{B}=\langle c^{\mathfrak{A}},c\rangle$.
\end{itemize}
Thus, no two constants denote the same element in $\mathfrak{B}$. Also, a routine induction on the structure of formulas shows that, for any equality-free formula $\phi$ with $m$ free variables, any $a_1,\dots,a_m\in A$, and any $c_1,\dots,c_m\in C$,
$$\mathfrak{A}\models\phi[a_1,\dots,a_m] \iff \mathfrak{B}\models\phi[\langle a_1,c_1 \rangle,\dots,\langle a_m,c_m \rangle].$$
Then, in particular, $\mathfrak{A}$ and $\mathfrak{B}$ satisfies the same equality-free sentences over the signature. This completes the proof.
\end{proof}

\begin{theorem}
\label{theorem_neq_cons}
Let $\phi$ be a first-order sentence in clausal form, possibly containing inequalities between distinct constants but no other occurrences of (in)equalities. Let $\phi'$ be the result of replacing all inequalities in $\phi$ by $\top$ (so $\phi'$ is equality-free). 
Then: for any equality-free first-order sentence $\psi$, $\models \phi\to\psi$ iff $\models\phi'\to\psi$.
\end{theorem}

\begin{proof}
First, the `if' direction holds since we obviously have $\models\phi\to\phi'$. For the other direction, suppose that $\models \phi\to\psi$ and $\mathfrak{A}\models\phi'$. Then, by Lemma~\ref{lemma_constat_inequalities}, there exists $\mathfrak{B}$ such that: (i) $\mathfrak{B}\models (a\foneq b\leftrightarrow\top)$ for any distinct constants $a$ and $b$, and (ii) $\mathfrak{B}$ and $\mathfrak{A}$ satisfy the same equality-free sentences. Since $\phi'$ is equality-free, we have by (ii) that $\mathfrak{B}\models\phi'$. Then, by (i) and the Replacement Theorem (see e.g. \cite[Thm. 8.2.1]{Fitting1996}), $\mathfrak{B}\models\phi$, and thus $\mathfrak{B}\models\psi$. Since $\psi$ is equality-free, we also have $\mathfrak{A}\models\psi$, which shows that $\models \phi'\to\psi$.
\end{proof}
In particular, this allows us to replace inequalities between (distinct) constants
by~$\top$ in the \BPLplus clauses returned by SCAN without any change to the logical consequences
in \BPL and the ordered fragment. Notice that replacing a disjunct in a clause by $\top$ amounts to deleting the clause altogether. Moreover, 

\begin{proposition}
\label{prop_neq_del}
If $M$ is a set of \BPLplus clauses returned by SCAN (on a set $N$ of $BPL$ clauses with fixed constant positions as input), and $M'$ is the result of deleting all clauses in $M$ containing inequalities, then $M'$ is a uniform interpolant (for $N$) in $BPL$.
\end{proposition}

\begin{proof}
Let $C\in M$ be a \BPLplus clause in which a variable $u$ is not prefix stable. Then there are two distinct constants $a$ and $b$ occurring in the same position in two prefixes of $u$. By the rules of constraint resolution,
it is easy to verify that $C$ also contains the constraint $a\not\approx b$, and then $C\notin M'$. Thus,
$M'$ is a finite set of clauses in $BPL$ (with fixed constant positions). Also, the correctness of SCAN and Theorem~\ref{theorem_neq_cons} ensures that $M'$ is a uniform interpolant for $N$.
\end{proof}

Therefore, at this stage, we have already shown that uniform interpolants in \BPL can be extracted from the result of SCAN. 
Note that this result is a little finer than Theorem~\ref{UI_BPL} because it establishes the uniform interpolation property within \BPL clauses where constants have fixed positions.

The remaining problem is to extract uniform interpolants \emph{in} the ordered fragment in this framework.  
We have seen in Section~\ref{section_OF_CSk} that each ordered sentence can be transformed into a set of $BPL$ clauses (with fixed constant positions) and such a transformation preserves logical consequences in the ordered fragment in both directions. 
Since $BPL$ clauses contain no functional terms, it is always possible to unskolemise the output of SCAN (after deleting clauses with inequalities) and get an (equality-free) first-order sentence in the $\exists^*\forall^*$-fragment.
The challenge we are facing here is that these sentences are in general not expressible in the ordered fragment. 
For example, let us start from the following ordered sentence as an input:
\begin{multline*}
    \forall x_1\forall x_2 \exists x_3((\neg X(x_1,x_2) \mlor P(x_1,x_2,x_3))\mland (\neg Y(x_1,x_2) \mlor Q(x_1,x_2,x_3))) \\
    \mland \ \forall x_1 \exists x_2 X(x_1,x_2)\mland \ \forall x_1\exists x_2 Y(x_1,x_2).
\end{multline*}
By constant Skolemisation and clausification, we have clauses 1--4 as below.
\begin{align*}
1.\ & \neg X(x,y) \mlor P(x,y,a)\\
2.\ & \hphantom{\neg} X(x,b)\\ 
3.\ & \neg Y(x,y) \mlor Q(x,y,a)\\
4.\ & \hphantom{\neg} Y(x,c)\\
\text{\sout{5.}}\ & P(x,y,a) \mlor x \foneq x' \mlor y \foneq b && \text{(C-res, 1, 2)}\\
5'.\ & P(x,b,a)  && \text{(2 $\times$ CEl)}\\
\text{\sout{6.}}\ & Q(x,y,a) \mlor x \foneq x' \mlor y \foneq c && \text{(C-res, 3, 4)}\\
6'.\ & \hphantom{\neg} Q(x,c,a)  && \text{(2 $\times$ CEl)}
\end{align*}
Applying rules as above, Clauses $5'$ and $6'$ remain, and we get the following first-order formula via the usual unskolemisation (i.e., each Skolem constant is eliminated by introducing an existential quantifier):
$$\exists u \exists v \exists w \forall x (P(x,u,w) \mland Q(x,v,w)).$$
However, it is not an ordered sentence, and we show below that it is not equivalent to any ordered sentence.

We say that a first-order sentence $\phi$ is \emph{invariant under ordered bisimulations} if, for any structures $\mathfrak{A}$ and $\mathfrak{B}$, the following implication holds:
\[\mathfrak{A}\sim\mathfrak{B}\ \Longrightarrow\ \mathfrak{A}\models\phi\iff\mathfrak{B}\models\phi.\]
The following result is included in \cite[Thm. 5]{Bednarczyk2022}. 
\begin{theorem}\label{thm:expressivity}
A first-order sentence is logically equivalent to an ordered sentence iff it is invariant under ordered bisimulations.
\end{theorem}

Let $\phi :=\exists u \exists v \exists w \forall x (P(x,u,w) \mland Q(x,v,w))$.
To show that this sentence is \emph{not} equivalent to any sentence of the ordered fragment, we only need to show that it is not invariant under ordered bisimulations. Let $P$ and $Q$ be ternary and $\mathfrak{A}$ the $\{P,Q\}$-structure (with domain $A$) such that:
\begin{align*}
    A&=\{a_1,a_2\}\\
    P^{\mathfrak{A}}&=\{(a_1,a_1,a_1),(a_2,a_1,a_1)\}\\
    Q^{\mathfrak{A}}&=\{(a_1,a_2,a_1),(a_2,a_2,a_1)\}.
\end{align*}
We can easily verify that $\mathfrak{A}\models\phi$. Now consider the structure $\mathfrak{A}'$ over the same domain $A$ such that:
\begin{align*}
    P^{\mathfrak{A}'}&=\{(a_1,a_1,a_1),(a_2,a_1,a_1)\}\\
    Q^{\mathfrak{A}'}&=\{(a_1,a_2,a_2),(a_2,a_2,a_2)\}.
\end{align*}
We observe that $\mathfrak{A}'\nvDash\phi$. 
Let $Z\subseteq \bigcup_{n\geq 0}(A^n\times A^n)$ be the relation that agrees with the identity relation on $\bigcup_{n\geq 0}A^n$ except that, for any $\bar{a}\in \bigcup_{n\geq 0}A^n$:
\begin{align*}
    &a_1a_2a_1\bar{a}Za_1a_2a_2\bar{a}\\
    &a_2a_2a_1\bar{a}Za_2a_2a_2\bar{a}.
\end{align*}
A routine check shows that $Z$ is an ordered bisimulation between $\mathfrak{A}$ and $\mathfrak{A'}$ and, in addition, $\varepsilon Z\varepsilon$. Thus, $\mathfrak{A}\sim\mathfrak{A}'$, and $\phi$ is not invariant under ordered bisimulations. By Theorem~\ref{thm:expressivity}, $\phi$ is not equivalent to any ordered sentence.

As SCAN tries to exhaust all consequences in first-order logic, it is not a surprise that it also produces formulas which fall beyond the expressivity of the ordered fragment. In order to compute uniform interpolants in the ordered fragment, we need to weaken the output set of clauses (so that it can be translated back to the ordered fragment) and meanwhile retain all of its consequences in the ordered fragment. At the time of writing, we do not know how this can be achieved.

\section{Conclusion}

In this paper, we have used the decidability result of unrefined
resolution
to show the uniform interpolation property of basic path logic. 
Also, we characterised the search space of the SCAN algorithm on \BPL clauses with fixed constant positions, proving that SCAN is terminating and therefore SOQE is solvable in this case. 
For computing uniform interpolants of ordered sentences, we proposed constant Skolemisation, with which the clausification of ordered sentences is proved to be consequence-preserving in the fragment (not just satisfiability-preserving). However, although constraints in the output of SCAN can always be eliminated and the resulting (equality-free) set of clauses can be expressed by an (equality-free) $\exists^*\forall^*$-formula, we note that it is not generally expressible in the ordered fragment. For future research we would like to find a method for extracting uniform interpolants in the ordered fragment from \BPL clauses.

\backmatter

\bmhead{Acknowledgments}
We would like to thank the reviewers and audience of the CADE-30
Workshop in Honor of Christoph Weidenbach’s 60th Birthday hed in August
2025, and participants of Dagstuhl Seminar 26091 on ``Revisiting the Foundations of Deduction in a New World'' held in February 2026,
where earlier versions of this work were presented.

\section*{Declarations}

\bmhead{Funding}

The research was not supported by grant funding.

\bmhead{Competing interests}

The authors have no conflicts of interest to declare that are relevant to the content of this article.

\bibliography{w60}

\end{document}